\documentclass[11pt]{article}

\usepackage{amssymb,amsmath}
\usepackage{enumerate}
\usepackage{ctable,subfig}
\usepackage{makecell}
\usepackage{mathrsfs}
\usepackage{verbatim}
\usepackage{nicematrix}
\usepackage{setspace}
\usepackage[linesnumbered,ruled]{algorithm2e}

\usepackage{graphicx}
\usepackage[normalem]{ulem}
\usepackage[left=2cm,top=2cm,right=2cm,bottom=2cm]{geometry}

\newtheorem{proposition}{Proposition}

\newtheorem{remark}{Remark}
\newtheorem{example}{Example}

\begin{document}
	\title{\bf Lower and upper bounds of the superposition of renewal processes and extensions}
	\author{{Shaomin Wu}\\
		\small{\it Kent Business School, University of Kent, Canterbury, Kent CT2 7FS, UK}}
	\date{}
	\maketitle

	\begin{center}{\bf Abstract}\end{center}
	Consider a system consisting of multiple sockets into each of which a component is inserted. If a component fails, it is replaced immediately and system operation resumes. Then the failure process of the system is the superposition of renewal processes (or superimposed renewal process, SRP). If the label of the components that cause the system to fail are unknown but the times between failures are known, we refer to such data as {\it masked failure data}. To estimate the SRP based on masked failure data is challenging. 

    This paper obtains the lower and upper bounds of the rate of the SRP when only masked failure data are available. If repair (rather than replacement) is conducted on failed components, the failure process of the system is the superposition of generalized renewal processes (SGRPs). The paper then derives the lower and upper bounds of the rate of SGRPs and proposes to use a weighted linear combination of the bounds to approximate the SGRP. Discussions are provided for possible extensions of the bounds for systems with other structures such as parallel systems. An algorithm for simulating the SGRP is then proposed. Numerical examples are used to illustrate the proposed approximation method.
	
	\noindent{\it Keywords}: Stochastic processes;  non-homogeneous Poisson process (NHPP); generalized renewal process (GRP); superimposed renewal process (SRP); superposition of generalized renewal processes (SGRP).
	
	\date{}
	\maketitle
	\section{Introduction}
\subsection{Motivation}
Consider a system consisting of multiple sockets into each of which a component is inserted. If a component fails, it is replaced immediately and the system operation resumes. The time on replacement is negligible. Then the failure process of the system is the superposition of renewal processes (or superimposed renewal process, SRP). To estimate the parameters in the SRP based on failure data, we need both failure data (times-between-failures) of each component and to know which components, i.e., the labels of the components, cause the system to fail. Assuming that the times between failures are available but which components cause the system to fail are unknown, then the failure data are referred to as {\it masked failure data}. It is not possible to obtain the SRP for the system if only masked failure data are available. Under such a scenario, one can only develop methods to approximate the SRP. 

It is known that most of the real-world technical systems are multi-component systems and are structured in series. It is often the case that only masked failure data from the fields such as warranty claim data are available for model estimation. Additionally, SRP has many real applications, including queueing systems in the healthcare sector \cite{he2020diffusion}, failure data modelling in reliability engineering \cite{wu2019failure,wu2020two}, and traffic analysis \cite{singh2017stochastic}. Hence, it is  important to obtain the SRP or develop a method to approximate the SRP. As such, in the literature, there are many authors studying the SRP \cite{cox1954superposition,khinchin1956poisson,drenick1960failure}, as reviewed in Section \ref{PriorWork}.

In reliability engineering, the term {\it repair} is a concept more widely used than {\it replacement}: a repair is usually assumed to bring a failed component to one of the three possible states:  {\it perfect}, {\it imperfect}, and {\it minimal}. An {\it perfect} repair functions the same as a replacement and restores the system to the as-good-as new status. A {\it minimal repair} implies that a failed item is restored to the status immediately before the occurrence of the failure or restores the system to the {\it as-bad-as-old} status. An {\it imperfect} repair restores the failed component to a status between the as-good-as-new status and the as-bad-as-old status. To model different levels of repair, stochastic processes are used. Usually, the renewal process (RP) is used to model the failure process of a system with perfect repair, the non-homogeneous Poisson process (NHPP) models that of a system with minimal repair, and a stochastic process that can model all the three situations is referred to as the generalized renewal process (GRP). Many GRP models have been developed in the literature, for example, the arithmetic reduction of age (ARA) model \cite{doyen2004classes}, the arithmetic reduction of intensity (ARI) model \cite{doyen2004classes}, the geometric process and its extensions \cite{wu2022double}. 

We may relax the assumption of the replacement in the SRP: We assume that a failed component is repaired instead of replaced, which is more close to the real practice as repair covers replacement. If the failure process of each component in the system is modelled by the GRP, then the failure process of the system is the superposition of generalized renewal processes (SGRP).

In the literature, there is little research on the SGRP. Obviously, the SRP is a special case of the SGRP. This paper therefore aims to obtain the lower and upper bounds of the rate of the SRP, extend the results to the SGRP situation, and then use a weighted linear combination of the bounds as an approximation of the SGRP.

It should be highlighted that the results of this paper can be utilized not only in the reliability engineering but also in other areas such as inventory management and neuro-physiology.
\subsection{Prior work and comments}
\label{PriorWork}
There are some work on the SRP. Cox and Smith \cite{cox1954superposition} proved that the observed failure times of an SRP tend to be distributed with a homogeneous Poisson process when the number of components goes to infinity and the time is far from the origin. Khinchin \cite{khinchin1956poisson} further clarified that the SRP
tends to be a non-homogeneous Poisson process. Drenick \cite{drenick1960failure} proved the same property as that of \cite{cox1954superposition} holds even if the failure processes of the multi-sockets follow heterogeneous renewal processes.  
 
In the applications of the SRP, \cite{kallen2011modelling} uses the SRP to model the effect of imperfect maintenance. \cite{bratiychuk2003application} applies the SRP to study batch arrival queues.

Some methods have been developed to approximate the SRP.  \cite{whitt1982approximating} proposes a comprehensive description of two basic methods for approximating the SRP: the stationary-interval method and the asymptotic method. Both methods determine the approximating renewal process by identifying moments for intervals between successive points and fitting a convenient distribution to the moments.
 \cite{albin1984approximating} proposes a hybrid method that combines two basic methods described by \cite{whitt1982approximating} and that approximate the complex superposition process by a renewal process. \cite{tortorella1996life} proposes a likelihood inference with a pooled discrete renewal process model. \cite{torab2001approximate} minimizes the mean-squared rate error, i.e., the difference between the rate of the superposition process and that of its renewal model. \cite{peixoto2009estimating} proposes a likelihood based method with random assignments of failure times onto components. \cite{zhang2017estimating} enumerate all possible patterns of missing component labels and propose a likelihood based method to estimate the component-wise renewal distribution under the assumption that all renewal distributions are identical among components. \cite{yamamoto2020exponential} extend the method proposed in \cite{zhang2017estimating} for the cases of heterogeneous components.  \cite{li2021estimating} proposes a nonparametric procedure to estimate the inter-occurrence time distribution by properly deconvoluting the renewal equation with the empirical renewal function. 
 
 It can be seen that most of the existing research focuses on the SRP. Little research investigates the SGRP. To the best of our knowledge, there is only one paper on the SGRP: \cite{wu2020two} proposes two methods to approximate the SGRP but does not provide the bounds of the rate of the SGRP. Although \cite{wu2017two,wu2019failure} develop methods to appropriate the GRP and then applies their appropriate methods the SGRP, neither of those papers provides a rigorous derivation process in obtaining their approximation methods. As such, this paper aims to bridge this gap.
\subsection{Proposed ideas}
\label{Ideas}
This paper aims to derive the lower and upper bounds of the rate of the SRP and  then extend the result to the case of the SGRP. The underlying ideas of the derivation can be interpreted as follows: the lower bound is the rate of the SGRP that repairs on failed components are assumed to be always performed on the oldest (or the less reliable) component in the system and the upper bound is the rate of the SGRP that repairs are assumed to be always performed on the youngest (or the most reliable) component in the system. 

An approximation method is then proposed based on the linear combination of the lower and upper bounds. We then propose a method to simulate the SGRP.
\subsection{Overview}
The remainder of the paper is structured as follows. Section \ref{sec:assumptions} lists assumptions and notations. Section \ref{sec:bounds} derives the lower and upper bounds. Section \ref{sec:approximation} proposes an approximation method that combines the lower and upper bounds and an algorithm to simulate the failure process of the SGRP. Section \ref{sec:SimulationExamples} offers numerical examples to illustrate the simulation algorithm proposed in Section \ref{sec:approximation}. Section \ref{sec:conclusions} concludes the paper.
\section{Notations and assumptions}
\label{sec:assumptions}
This section sets notations, as shown in Table \ref{notations}, and assumptions, both of which will be used in this paper.

The notation table is shown in \ref{notations}.
\begin{table}[!ht]
\caption{Notations} 
\begin{tabular}{rp{14.5cm}} \hline
	$n$ & number of components in a series system;\\
	$i$ & index: number of components;\\
	$j$ & index: order of memory in some models, e.g., an ARA (arithmetic reduction of age) or ARI (arithmetic reduction of intensity) model; \\
	$k$ & index: number of failures (or repairs) of the system or a component;\\
	$t$ & time since the system starts; the components in the system at $t=0$ are new; \\
	$\mathscr{H}_{i,t-}$&  history of failures of component $i$ up to time $t$;\\
    $\mathscr{H}_{s,t-}$&  history of failures of the system up to time $t$;\\
	$T_i^{(k)}$  & time when the $k$th failure of component $i$ occurs;\\
	$T_s^{(k)}$  & time when the $k$th failure of the system occurs;\\
	$t^{(k)}$  & observations of $T^{(k)}$;\\
	$N_{i,t}$, $N_i(t)$ & number of failures that  component $i$ has experienced up to time $t$;\\
	$N_{s,t}$, $N_s(t)$ & number of failures that the system has experienced up to time $t$;\\
	$\lambda_s(t)$ & failure intensity function of the system at time $t$; \\  
	$\lambda_s^{(k)}(t)$ & failure intensity function of the system that has experienced $k$ failures to $t$; \\  	
	$\widetilde{\lambda}_s(t)$ & approximation function to  $\lambda_s(t)$; \\    
	$\lambda_t^{(k)}$, $\lambda^{(k)}(t)$ & failure intensity function of a component at $t$ and the system has experienced $k$ failures; \\
	$\lambda(t)$ & failure intensity function of a component at time $t$ before the first failure; Note $\lambda(t)=\lambda^{(0)}(t)$\\
	$\lambda_{i,t}^{(k)}$, $\lambda_i^{(k)}(t)$ & failure intensity function of  component $i$ at $t$ and the system has experienced $k$ failures. \\ \hline
	\label{notations}
\end{tabular}
\end{table}

The paper uses the following assumptions.
\begin{description}
	\item[A1.] Suppose a series system is composed of $n$ identical components whose failures are statistically independent.
	\item[A2.] The failure rate or the initial failure intensity of a component is $\lambda(t)$ before its first failure, where $\lambda(t)$ increases in $t$.
	\item[A3.] Once a failure occurs, it is  immediately repaired. The time on the repair is negligible.
	\item[A4.] The repair effectiveness may be depicted by a model such as the virtual age model with fixed parameters, perfect, imperfect, or minimal (PIM). For example, if a virtual age model is used, we can assume that $V_k=V_{k-1}+A X_k$, where $V_k$ is the virtual age of the item after the $k$th repair, $X_k$ is the lifetime of the item after the $k$th repair and $A$ is the degree of the repair. 
\end{description}
\section{Lower and upper bounds of the rate of a SGRP}
\label{sec:bounds}
Let $\{T_i^{(k)}: k=1,2,...\}$ be the successive failure times of component $i$ (hence $T_i^{(k+1)} > T_i^{(k)}$), starting from $T_{i,0}=0$, and $N_{i,t}$ be the number of failures up to time $t$. Let $\mathscr{H}_{i,t-}$ denote the history of the failure process up to $t$ (exclusive of $t$). The failure process of the component can be defined equivalently by the random processes $\{T_i^{(k)}\}_{k \ge 1}$ or $\{N_{i,t}\}_{t \ge 0}$ (i.e., $\{N_{i}(t)\}_{t \ge 0}$) and is characterised by the intensity function,
\begin{equation}
	\lambda_i^{(N_{i,t})}(t|\mathscr{H}_{i,t-})=\lim_{\Delta t \downarrow 0} \frac{P\{N_i(t+\Delta t)-N_i(t) = 1 |\mathscr{H}_{i,t-}\}}{\Delta t},
	\label{eq:Intensity}
\end{equation}
where $P\{N_i(t+\Delta t)-N_i(t) = 1 |\mathscr{H}_{i,t-}\}$ is the probability that component $i$ fails within the interval $(t, t+\Delta t)$, given the history of failures of the component up to time $t$, $\mathscr{H}_{i,t-}$. 

Similarly, we can define the failure intensity of the system as $\lambda_s^{(N_{s,t})}(t|\mathscr{H}_{s,t-})$ by replacing $\lambda_i^{(N_{i,t})}(t|\mathscr{H}_{i,t-})$ in Eq. \eqref{eq:Intensity} with $\lambda_s^{(N_{s,t})}(t|\mathscr{H}_{s,t-})$, $\mathscr{H}_{i,t-}$ with $\mathscr{H}_{s,t-}$, $N_{i,t}$ with $N_{s,t}$, and $T_i^{(k)}$ with $T_s^{(k)}$, respectively, where $N_{s,t}$ or $N_s(t)$ is the number of failures of the system that has been observed up to time $t$ and $\mathscr{H}_{s,t-}$ is the history of failures of the system up to time $t$, $N_{s,t} \ge 0$.

Although the failure intensity function of an item (which may be a component or the system) should be denoted by the memory of $\mathscr{H}_{s,t-}$ such as $\lambda_i^{(N_{i,t})}(t|\mathscr{H}_{i,t-})$ and $\lambda^{(N_{s,t})}_s(t|\mathscr{H}_{s,t-})$. For simplicity, this paper will use notations $\lambda^{(i,N_{i,t})}(t)$ and $\lambda_s^{(N_{s,t})}(t)$, respectively.

 The paper uses the term the ``failure intensity function'' of the SRP or SGRP and the term the ``rate'' of the SRP or SGRP exchangeably.
\subsection{Lower and upper bounds of the rate of SRP with identical components}
\label{sec:bounds-SRP}
In this section, we derive the lower and upper bounds of the rate of the SRP.

Suppose a system consisting of $n$ sockets into each of which a component is inserted. If a component fails, it is replaced immediately and system operation resumes. As above discussed, the failure process of the system is the superimposed renewal processes. We have the following proposition.
\begin{proposition}\label{prop:proposition-1} If replacement is performed on each failed component, or the failure process of the system is the SRP, then
	\begin{equation}
			\sum_{i=0}^{n-1}  \lambda(t-T_s^{(N_{_{s,t}}-i)})\le \lambda_s(t) \le (n-1)\lambda(t)+ \lambda(t-T_s^{(N_{s,t})}),
			\label{eq:proposition-1}
	\end{equation}  
where $T_s^{(k)}=0$ if $k \le 0$ and $N_{s,t} \ge 0$.
\end{proposition}
\begin{proof} It is noted that neither $N_{_{i,t}}$ nor $T_i^{(N{_{i,t}})}$ is available. 

Since a replacement is carried out at time $T_i^{(N{_{i,t}})}$ on component $i$, then 
\begin{align}
\lambda_i^{(N{_{i,t}})}(t)=\lambda(t-T_i^{(N{_{i,t}})}),
\label{eq:perfect}
\end{align}
where $\lambda_i^{(N{_{i,t}})}(t)$ is the failure rate of component $i$ at time $t$, and $\lambda(t-T_i^{(N{_{i,t}})})$ is the failure rate of component $i$ at time $t$ and  its latest replacement occurs at time $T_i^{(N{_{i,t}})}$.

Component $i$ becomes more reliable after a replacement at time $T_i^{(N{_{i,t}})}$ than without the replacement at that time point, then, 
\begin{align}
\lambda_i^{(N{_{i,t}})}(t) \le \lambda_i^{(N_{_{i,t}}-1)}(t).
\label{eq:minimal}
\end{align}

Since the system is structured in series, the failure intensity function of the system can be expressed by
\begin{align}
 \lambda_s(t)=\lambda_s^{(N_{s,t})}(t) = \sum^n_{i=1} \lambda_i^{(N{_{i,t}})}(t).
\label{eq:system-repair}
\end{align}
Now we are going to prove the first part of the inequality, $\displaystyle{\sum_{i=0}^{n-1}  \lambda(t-T_s^{(N_{_t}-i)})\le \lambda_s(t)}$, by induction.

If $N_{s,t}=0$, $\displaystyle{\lambda_s(t) = \sum^{n-1}_{i=0} \lambda_i^{(0)}(t) = \sum^{n-1}_{i=0} \lambda(t)}$. 

If $N_{s,t}=1$, which means the system has experienced one failure and the label of the component that causes the system to fail is unknown, then 
\begin{flalign}
 \lambda_s(t)&=\lambda_s^{(1)}(t) &\nonumber \\
 & = \sum^n_{i=1,i \neq j} \lambda_i^{(0)}(t) + \lambda_j^{(1)}(t)&  [\footnotesize{\textrm{assuming component $j$ has failed}}]\nonumber \\
 & = \sum^n_{i=1,i \neq j} \lambda(t) + \lambda(t-T_j^{(1)}) & [\footnotesize{\textrm{according to Eq. \eqref{eq:perfect}}}] \nonumber \\
 & = \sum^n_{i=1,i \neq j} \lambda(t) + \lambda(t-T_s^{(1)}). & [\footnotesize{\textrm{since there is only one failure, } T_{j,1}=T_1}] \nonumber \\
 & = \sum^{n-1}_{i=0} \lambda(t-T_s^{(N_{s,t}-i)}). & [\footnotesize{\textrm{because $T_s^{(N_{s,t}-i)}=0$ for } {N_{s,t}-i} \le 0}]
\label{eq:system-repair1}
\end{flalign}
In Eq. \eqref{eq:system-repair1}, the sentence on the right hand side of each equality is the  comment on the reason that the respective equality is derived.

Now assume the inequality holds when $N_{s,t}=N$. That is, the system has experienced $N$ failures up to time $t$ and 
\begin{equation} 
\lambda_s(t) \ge \sum_{i=0}^{n-1}  \lambda(t-T_s^{(N-i)})
\label{eq:N_induction}
\end{equation} 
holds. 

When $N_{s,t}=N+1$, we have
\begin{align}
\sum_{i=0}^{n-1}  \lambda(t-T_s^{(N+1-i)}) &= \lambda(t-T_s^{(N+1)})+\lambda(t-T_s^{(N)})+\dots \lambda(t-T_s^{(N-n+2)}) \nonumber \\
&= \sum_{i=0}^{n-2}  \lambda(t-T_s^{(N-i)}) + \lambda(t-T_s^{(N+1)}) \nonumber \\
&= \sum_{i=0}^{n-1}  \lambda(t-T_s^{(N-i)})-\lambda(t-T_s^{(N-n+1)}) + \lambda(t-T_s^{(N+1)})& \nonumber \\
& \le \sum_{i=0}^{n-1}  \lambda(t-T_s^{(N-i)})\qquad \qquad \qquad [\footnotesize{\textrm{Because } T_s^{(N-n+1)} \le T_j^{(N+1)} \textrm{and } \lambda(t) \textrm{ increases in } t}] \nonumber \\
&=\lambda_s(t)
\label{eq:Part1}
\end{align}
Now we are to prove the second part $\lambda_s(t) \le (n-1)\lambda(t)+ \lambda(t-T_s^{(N_{s,t})})$ as follows.

From inequality \eqref{eq:minimal}, we can obtain
\begin{align}
\lambda_i^{(N_{_{i,t}})}(t) \le \lambda_i^{(N_{_{i,t}}-1)}(t) \le  \lambda_i^{(N_{_{i,t}}-2)}(t) \le \dots \le \lambda_i^{(0)}(t)=\lambda(t).
\label{eq:minimal-1}
\end{align}
Hence, combining inequalities \eqref{eq:perfect}, \eqref{eq:system-repair},  and \eqref{eq:minimal-1}, we obtain
\begin{equation}
    \sum_{i=1}^n \lambda(t-T_s^{(N_{_t}-i)})\le \lambda_s(t) \le n \lambda(t).
    \label{eq:SRP1}
	\end{equation} 
Since the SRP assumes that repair on a failed component is perfect, it ensures that the component with the latest repair (i.e., the replacement at time $T_s^{(N_{s,t})}$) has failure rate $\lambda(t-T_s^{(N_{s,t})})$. Hence, we can obtain
\begin{equation}
 \lambda_s(t) \le (n-1) \lambda(t) + \lambda(t-T_s^{(N_{s,t})}).
    \label{eq:Part2}
\end{equation} 
Combining inequalities \eqref{eq:Part1} and \eqref{eq:Part2}, we establish Proposition \ref{prop:proposition-1}. \hfill{$\blacksquare$}
\end{proof}

Proposition \ref{prop:proposition-1} can be interpreted as follows. 

Always replacing the oldest component with a new component upon a failure of the system results in the largest reduction in the rate of the SRP, which keeps the system at the most reliable state and therefore the smallest failure intensity. On the other hand, always replacing the youngest component with a new component upon a failure of the system results in the smallest reduction in the rate of the SRP, which keeps the system at the least reliable state and therefore the largest failure intensity.

\subsection{Lower and upper bounds of the rate of SGRP with identical components}
\label{sec:bounds-SGRP}
Similar to Proposition \ref{prop:proposition-1}, for the case where failed components are repaired (instead of replacement), a similar proposition is given below.
\begin{proposition}\label{prop:proposition-2} Under assumptions A1--A4 in Section \ref{sec:assumptions}, the failure intensity $\lambda_s(t)$ of the system after the $N_{s,t}$-th repair satisfies
	\begin{equation}
  \sum_{i=0}^{n-1}  \lambda_i^{(N_{_t}-i)}(t) \le \lambda_s(t) \le (n-1)\lambda(t)+\lambda^{(N_{s,t})}(t),
			\label{eq:proposition-2}
	\end{equation}   
	where $\lambda_i^{(N_{_t}-i)}(t)=\lambda(t)$ for $N_{_t}-i \le 0$.
\end{proposition}
\begin{proof}
 The proof can be easily established by mimicking the proof process of Proposition \ref{prop:proposition-2}. 
\end{proof} 

Proposition \ref{prop:proposition-2} can be interpreted similarly to that of Proposition  \ref{prop:proposition-1}, but interpreted in another way, as follows. The left and right terms in Inequality \eqref{eq:proposition-2}  are two extreme scenarios: the most and the least reliable situations. In the most reliable situation, the failures of the system are always due to the failure of the most anciently repaired component; in the most unreliable situation, the failures of the system are always due to the failure of the most recently repaired component. These two scenarios form the lower bound and the upper bound of the failure intensity of the system, respectively.
\subsection{A series system composed of heterogeneous components}
Proposition \ref{prop:proposition-2} derived the lower and upper bounds of the rate of the SGRP based on the assumptions A1--A4, where A1 assumes  that the system is composed of identical components. If we assume the components in the system are heterogeneous, then we can derive the following upper bound.
\begin{proposition}\label{prop:proposition-3} Suppose a system is composed of $n$ components. The initial intensity functions of the components are $\xi_{i}(t)$ where $i=1, 2, \dots, n$ and $\xi_1(t)\le \xi_2(t) \le \dots \le \xi_n(t)$.  Assume the maintenance effectiveness upon failures is the same for each individual component. The failure intensity $\lambda_s(t)$ of the system after the $N_{s,t}$-th repair satisfies
	\begin{equation}
			\lambda_s(t) \le \sum_{i=2}^n\xi_i(t)+\xi_1^{(N_{s,t})},
			\label{eq:proposition-3}
	\end{equation}   
where $\xi_1^{(N_{s,t})}$ is the failure intensity function of component $i$ that has experienced $N_{s,t}$ failures.
\end{proposition}
\begin{proof} Similar to the proof of Proposition \ref{prop:proposition-1}, we can establish Proposition \ref{prop:proposition-3}.
\end{proof}
\subsection{Systems with other structures}
The preceding subsections investigate the lower and upper bounds of the rate of a SGRP. For a system that is structured in other types such as parallel systems or $k$-out-of-$n$ systems, the lower and upper bounds of the rate of the system failure process can be derived based on the same principle as in Proposition \ref{prop:proposition-2}. For example, for a parallel system composed of identical components, the lower bound of the failure intensity of the system can be derived by assuming that the component with the smallest conditional failure intensity is repaired, and the upper bound can be obtained by assuming that the component with the largest conditional failure intensity is repaired.
\section{Approximation of the SGRP}
\label{sec:approximation}
This section aims to propose a model to approximate the SGRP and then propose an algorithm to simulate the approximation method.
\subsection{A model to approximate the SGRP}
An interesting question is to approximate the rate of the SGRP. Based on Proposition \ref{prop:proposition-2}, we may use the following model, which is the weighted average of the lower bound and the upper bound of the Inequality \eqref{eq:proposition-2}, respectively, to approximate the rate of the SGRP of the $n$ components in a series system for $N_{s,t} \ge 1$. 
	\begin{align}
		\widetilde{\lambda}_s(t) &=
		\delta \sum_{i=0}^{n-1} \lambda^{(N_{_t}-i)}(t)+(1-\delta)\left( (n-1)\lambda(t)+\lambda^{(N_{s,t})}(t)\right), \nonumber \\
		&=(n-1)(1-\delta)\lambda(t) + (1-\delta) \lambda^{(N_{s,t})}(t) + \delta\sum_{i=0}^{n-1}  \lambda^{(N_{_t}-i)}(t)
		\label{eq:MyModel-0}
\end{align} 
where $\delta \in [0,1]$.

\begin{remark} Eq. \eqref{eq:MyModel-0} has the following special cases.
	\begin{itemize}
		\item If $\lambda(t)=\lambda$, then $ \widetilde{\lambda}_s(t)=\lambda$. That is, the failure intensity of a system composed of components with a constant failure intensity, i.e., failure rate, is constant.
		\item If $n=1$, then $ \widetilde{\lambda}_s(t)=\lambda^{(N_{s,t})}(t)$. That is, the system is composed of a single component, whose repair process is modelled by a GRP.
		\item If $\delta=0$, then  $ \widetilde{\lambda}_s(t)= (n-1)\lambda(t)+\lambda^{(N_{s,t})}(t)$. This implies that the system can be assumed to be a two-component system, in which minimal repair is applied to the component with intensity function $(n-1) \lambda(t)$ and PIM is applied to the component with failure intensity function $\lambda^{N_{s,t}}(t)$. This case is also a variant of Model I proposed in \cite{wu2017two}, in which perfect repair is applied to the component with failure intensity function $\lambda^{N_{s,t}}(t)$.
		\item If $\delta=1$, then model \eqref{eq:MyModel-0} is the failure intensity of a system composed of $n$ components in series with the failure process as a SRP on each repair and reduces to the  moving average of intensity model (MAI) proposed in \cite{wu2019failure}, which shows the outstanding performance compared with ten other models on 15 real datasets and simulation data in terms of AIC (Akaike information criterion), BIC (Bayesian information criterion) and corrected AIC.
		\item If $\delta=0.5$, then model \eqref{eq:MyModel-0} reduces to Model II proposed in \cite{wu2017two}. According to \cite{wu2017two}, Model II outperforms five other models on simulated data and a real dataset in terms of the AIC.
	\end{itemize}
\end{remark}

Denote $\lambda_{s,L}(t)=\sum_{i=0}^{n-1} \lambda^{(N_{_t}-i)}(t)$ and $\lambda_{s,U}=(n-1)\lambda(t)+\lambda_t^{(N_{s,t})}$.  $\lambda_{s,L}(t)$ is the failure intensity of the system whose failures are always due to the failures of the most anciently repaired component, or the failure of the oldest component. $\lambda_{s,U}(t)$ is the failure intensity of the system whose failures are  always due to the failures of the component that has been most recently repaired, or the youngest component. The approximation method in Eq. \eqref{eq:MyModel-0} is a weighted average of two extremist scenarios. Below is an example that provides a visual interpretation. 

\begin{example} Given a series system that consists of four components, whose failures at time points are shown in the top four horizontal lines in Figure \ref{fig:fig01}. The superposition of the four failure processes is shown at the bottom horizontal line: $T_s^{(1)}, T_s^{(2)}, \dots, T_s^{(11)}$ . As can be seen, components 1, 2, 3 and 4 have 3, 2, 2, and 4 failures, respectively. 
	\begin{figure}[h]
		\centering
		\includegraphics[width=0.9\linewidth]{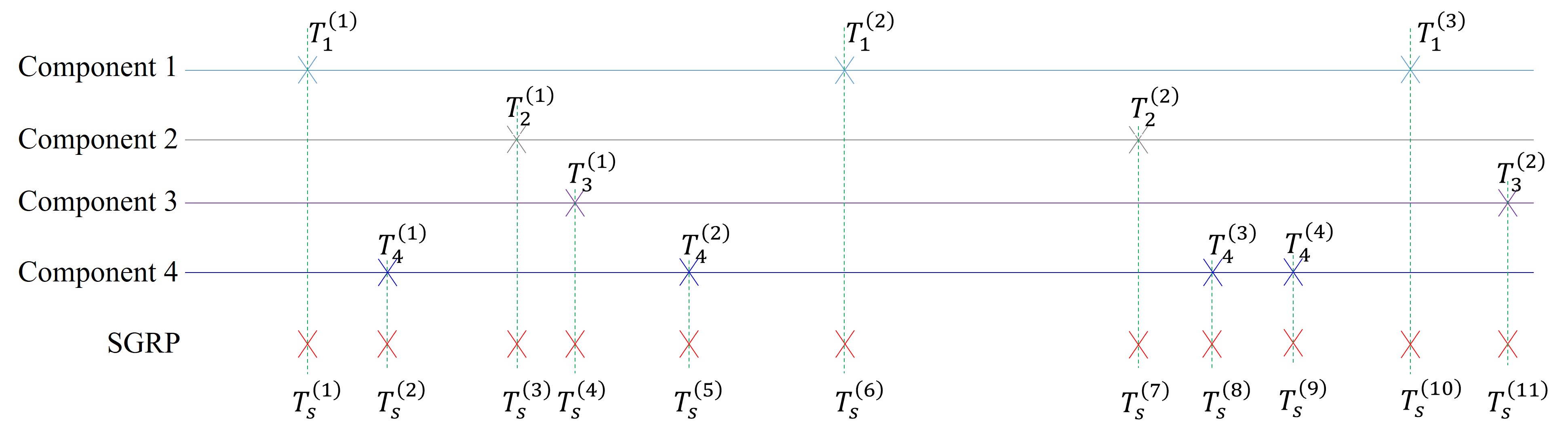}
		\caption{An example of the superposition of four imperfect repair processes.}
		\label{fig:fig01}
	\end{figure}
	Now let us assume that only masked failure data are available. That is, only values of $T_s^{(1)}, T_s^{(2)}, \dots, T_s^{(11)}$ are available but $T_i^{(k)}$ are unavailable, as illustrated in Figure \ref{fig:fig02}.
	\begin{figure}[h]
		\centering
		\includegraphics[width=0.9\linewidth]{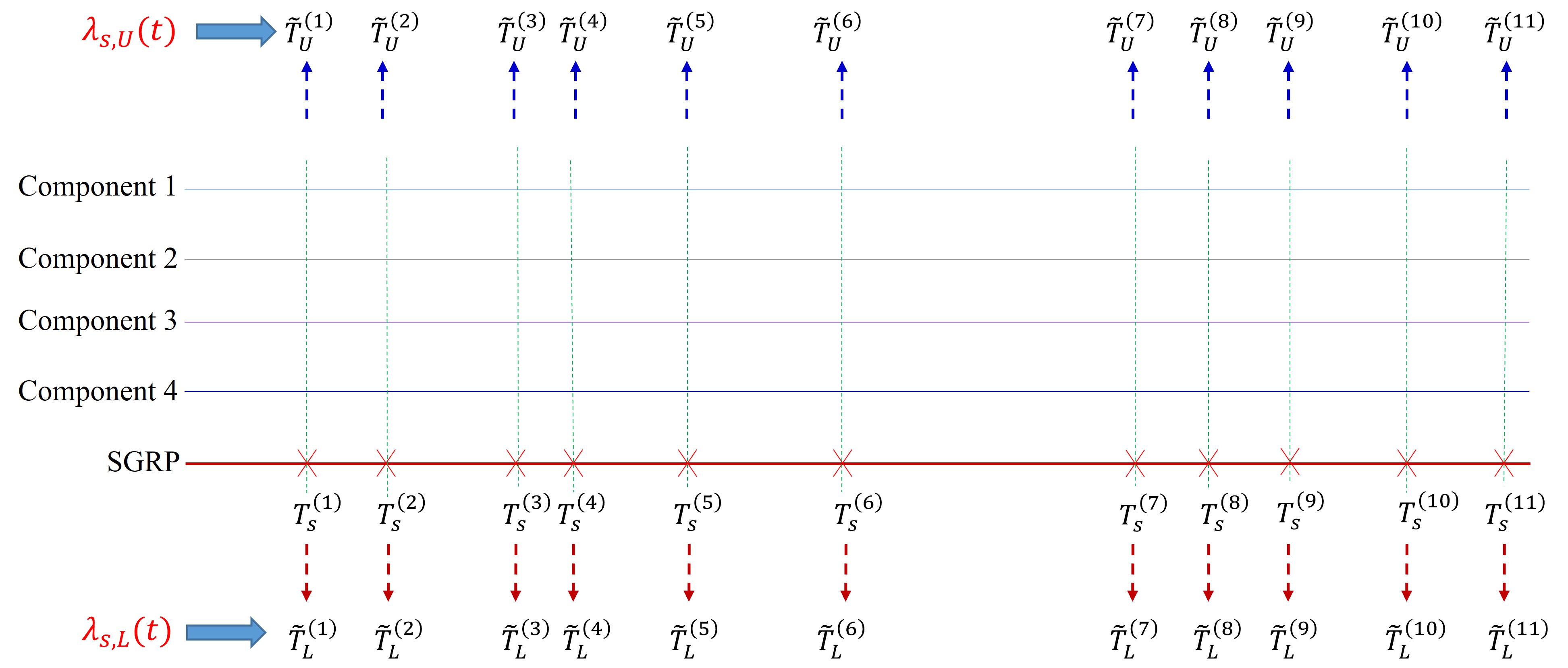}
		\caption{$\lambda_{s,L}(t)$ and $\lambda_{s,U}(t)$ assume the masked failure data to be unmasked ones, respectively.}
		\label{fig:fig02}
	\end{figure}
\end{example}
In the literature, there are many GRPs that have been proposed. Below we assume the arithmetic reduction of age (ARA) model \cite{doyen2004classes} is applied. 
\begin{example}
In case the arithmetic reduction of age (ARA) model \cite{doyen2004classes} is used to model the repair process of component $i$, then its initial failure intensity function is given by
\begin{equation}
		\lambda_i^{(N_{i,t})}(t)=\lambda\left(t- \rho\sum_{j=0}^{\min\{m-1,N_{_{i,t}}-1\}} (1-\rho)^j T_{N_{_{i,t}}-j}\right),
		\label{eq:ARA-m}
\end{equation}   
where $\rho \in [0,1]$: (a) the repair is a harmful repair if $\rho<0$; (b) the repair is a minimal repair if $\rho=0$; (c) the repair is an good-as-new repair (or replacement) if $\rho=1$; and (d)  the repair is efficient if $\rho \in (0,1)$.

If $m=1$, then Eq. \eqref{eq:ARA-m} reduces to the virtual age model I \cite{kijima1989some}.

If the failure data are masked, the failure intensity, $\widetilde{\lambda}_s(t)$, of the superposition of the $n$ ARA$_m$ failure processes is therefore approximated by the expressions of the model \eqref{eq:MyModel-0}, where the repair process $\lambda^{(i,N_{s,t})}(t)$ follows the ARA$_m$ model.

According to  Eq. \eqref{eq:MyModel-0}, we have
\begin{itemize}
 \item when $N_{s,t} =0$, $ \widetilde{\lambda}_s(t)=\lambda(t);$
 \item when $1 \le N_{s,t} \le n$ and $n \ge 2$, $
		\widetilde{\lambda}_s(t)=(n-N_{s,t})\delta+(n-1)(1-\delta) \lambda(t) +(1-\delta) \lambda \left(t- \rho T_s^{(N_{s,t})}\right)+ \delta \sum_{i=1}^{N_{s,t}}  \lambda\left(t- \rho T_i\right);
$ and 
   \item when $N_{s,t} > n$, $
		\widetilde{\lambda}_s(t)
		=  (n-1)(1-\delta) \lambda(t) + (1-\delta) \lambda\left(t- \rho\sum_{j=0}^{q_{_{N_s(t)}}}(1-\rho)^j T_{N_{s,t}-j}\right)$	\\  $+\delta \sum_{i=0}^{n-1} \lambda \left(t- \rho \sum_{j=0}^{q_{_{N_s(t)}}} (1-\rho)^j T_{N_{s,t}-nj-i}\right),
$
where $q_{_{N_s(t)}}=\min\{\lfloor{\frac{N_s(t)}{n}}\rfloor -1,m-1\}$.
\end{itemize}
\end{example}
Then we can re-write model \eqref{eq:MyModel-0} as follows.
\begin{equation}
		\psi(N_{s,t},\lambda(t))=(\max\{n-N_{s,t},0\} \delta+(n-1)(1-\delta)) \lambda(t) + (1-\delta) \lambda_t^{(N_{s,t})} + \delta\sum_{i=0}^{\min\{N_{s,t}-1,n-1\}} \lambda_t^{(N_{_t}-i)}.
\end{equation} 
Then, model \eqref{eq:MyModel-0} can be re-written by 
	\begin{equation}
		\widetilde{\lambda}_{s}(t)= \left\{ \begin{array}{ll}
			\lambda(t) & \textrm{if $N_{s,t}=0$}\\
			\psi(N_{s,t},\lambda(t)) & \textrm{if $N_{s,t} \ge 1$}.
		\end{array} \right.
		\label{eq:MyModel-1}
	\end{equation}

\subsection{A simulation algorithm}
\label{sec:simulation}
To develop an algorithm to simulate the behaviour of the stochastic process model is useful in the real applications. This section therefore aims to develop an algorithm to simulate the failure process based on $\widetilde{\lambda}_s(t)$ in Eq. \eqref{eq:MyModel-0} and then shows some simulation examples.

Algorithm \ref{Simulation} shows the procedure of simulating the conditional failure intensity of $\widetilde{\lambda}_s(t)$ given in Eq. \eqref{eq:MyModel-0}, in which we assume that the failure rate of each component is $\frac{1}{n}\lambda(t)$. To simplify notations in the Algorithm, we denote $\Phi_1=\frac{\delta}{n} \sum_{i=0}^{\min\{N_{s,t}-1,n-1\}} \lambda^{(N_{_t}-i)}(t) $,  $\Phi_2=\frac{(1-\delta)(n-1)}{n} \lambda(t)$, and $\Phi_3=\frac{1-\delta}{n} \lambda_t^{(N_{s,t})}$, which corresponds the three items in Eq. \eqref{eq:MyModel-0}.

Basically, the simulation method proposed in Algorithm \ref{Simulation} regards the model in Eq. \eqref{eq:MyModel-0} as a series system that is composed of three subsystems with failure intensity functions $\Phi_1, \Phi_2$, and $\Phi_3$, respectively. See the descriptions after symbols `/*' and `*/'
in the Algorithm.

{\small
	\begin{singlespace}
		\begin{algorithm}[H]
			\KwData{\begin{itemize}
					\item Given $\Lambda_0(v)=\int_0^v \lambda(t) dt$ and $\Lambda_1(v)=\frac{(1-\delta)(n-1)}{n} \Lambda_0(v)$.
					\item Given $\tau_{i,k}'$ \tcc*{which are $n+1$ series of successive failure times that have already been generated based on a given base repair model $\lambda_t^{(N_{s,t})}$ with initial failure intensity function $\frac{\delta}{n}\lambda(t)$ and each series have $N$ data, where $i=1,2, \dots, n, n+2$ and $k=1,2, \dots, N$. That is, $\tau_{i,k}'$ is the $k$ failure time point of the $i$th component.}
			\end{itemize}}
			\KwResult{Simulated data: $t^{(1)}, t^{(2)}, \dots, t^{(N)}$.}
			Sort $\tau_{i,k}'$ according to the first column (i.e., $k=1$ and $i \in \{1,2,\dots,n\}$) in ascending order and denote the sorted matrix as $\tau_{i,k}$\;
			$s \gets 0$\;
			$\tau^*_0  \gets  0$\;
			$\tau_0  \gets  0$\;
			\tcc{\small The For-loop below is to generate data based on $\Phi_3$, or simulates an NHPP with cumulative intensity function $\Lambda_1(t)$}
			\For{$k=0; k \le N$}{
				Draw $s_{k+1} \sim U(0,1)$\;
				$\tau^*_{k+1}  \gets  \tau^*_{k}  -\log(s_{k+1})$\;
				$\tau_{n+1,k+1}  \gets  \inf\{v: \Lambda_1(v) \ge \tau^*_{k+1}\}$\;
			}
			$i_1  \gets 1$\;
			$i_2  \gets  1$\; $t^{(1)}  \gets  \min\{\tau_{1,1},\tau_{n+1,1},\tau_{n+2,1}\}$\;
			\tcc{The following For-loop generates random numbers based on $\Phi_1 + \Phi_2+ \Phi_3$ for  $N_{s,t} \le n$}
			\For{$2 \le i_1 \le n$}{
				\eIf{$t^{(i_1)}= \tau_{n+1,i_2}$}{
					$i_2  \gets i_2+1$\;  $\displaystyle{t^{(i_1)}  \gets  \min\{\tau_{i_1,1},\tau_{n+1,i_2},\tau_{n+2,i_1}\}}$\;
				}{
					$\displaystyle{t^{(i_1)}  \gets  \min\{\tau_{i_1,1},\tau_{n+1,i_2},\tau_{n+2,i_1}\}}$;
				}
			}
			\tcc{The following For-loop generates random numbers based on $\Phi_1 + \Phi_2+ \Phi_3$ for  $N_{s,t} > n$}
			\For{$i_1 \in \{n+1, \dots, N\}$}{
				$S_0  \gets  \{1,2,\dots,n\}$\;
				\For{$i_0 \in S_0$}{
					$t_{i_0}'  \gets  \tau_{i_0,i_1-i_0}$\;
					$i_2  \gets  i_2+1$\;  \eIf{$t^{(i_1)}= \tau_{n+1,i_2}$}{$\displaystyle{t^{(i_1)}  \gets \min_{i_0 \in S_0}\{t_{i_0}',\tau_{n+1,i_2},\tau_{n+2,i_1}\}}$; }{ $\displaystyle{t^{(i_1)}  \gets  \min_{i_0 \in S_0}\{t_{i_0}',\tau_{n+1,i_2},\tau_{n+2,i_1}\}}$\;
					}
				}
			}
			\caption{Simulation of Model \eqref{eq:MyModel-0}}\label{Simulation}
		\end{algorithm}
	\end{singlespace}
}
\section{Simulation examples}\label{sec:SimulationExamples}
Figures \ref{fig:SGRP}, \ref{fig:MyModel}, \ref{fig:SGRP_MyModel1} and Fig \ref{fig:SGRP_MyModel2} show the simulation results, where $\lambda(t)=\frac{1.3}{40}\left(\frac{t}{40}\right)^{0.3}$, the number of components in a system is $n$, and the number of failures is set to be $N=200,000$. Each curve in the figures illustrates the pattern of failure rates. The failure rate is defined as the ratio of the number of failures in 1000 units of time to 1000. In all of the figures, the failure process $\lambda_t^{(k)}$ is assumed to be based on ARA$_1$ from Eq. \eqref{eq:ARA-m}.

Figure \ref{fig:SGRP} shows three cases when the causes of failures are available, i.e., the failure data are unmasked. It is constructed based on Eq. \eqref{eq:ARA-m} in which $n=100$ and the values of $\rho$ are 0.3, 0.6, and 0.9, respectively. Figure \ref{fig:MyModel} shows six cases that are generated based on the algorithm in Table \ref{Simulation}. $n=100$ in all curves in this figure. By comparing Figures \ref{fig:SGRP} and \ref{fig:MyModel},  one can see that they share similar patterns.  From Figure \ref{fig:MyModel}, it can be seen that the failure rates  increase when $\delta$ decreases.

The two upper curves and the two lower curves in Fig. \ref{fig:SGRP_MyModel1} and Fig. \ref{fig:SGRP_MyModel2} are generated by setting $\delta=0$ and $\delta=1$ in Eq. \eqref{eq:MyModel-0}, respectively, and they correspond to the lower bounds and upper bounds, respectively. The two curves between the lower and upper curves in Figure \ref{fig:SGRP_MyModel1} and Figure  \ref{fig:SGRP_MyModel2} are generated by the SGRP, based on Eq. \eqref{eq:ARA-m}.  $\rho=0.3$ in Figure \ref{fig:SGRP_MyModel1} and $\rho=0.4$ in Figure  \ref{fig:SGRP_MyModel2}. It can be seen that the rates from Eq. \eqref{eq:ARA-m} are smaller than the upper bounds and larger than the lower bounds, respectively,  which confirms Proposition \ref{prop:proposition-2}, respectively.
\begin{figure}[!ht]
	\includegraphics[width=0.6\linewidth]{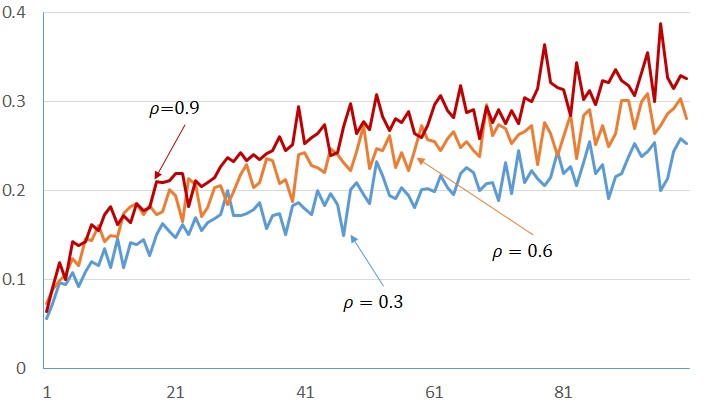}
	\caption{\; The three curves are generated based on the SGRP model with its failure intensity function shown in Eq. \eqref{eq:ARA-m}.}\label{fig:SGRP}
\end{figure}
\begin{figure}[!ht]
	\includegraphics[width=0.6\linewidth]{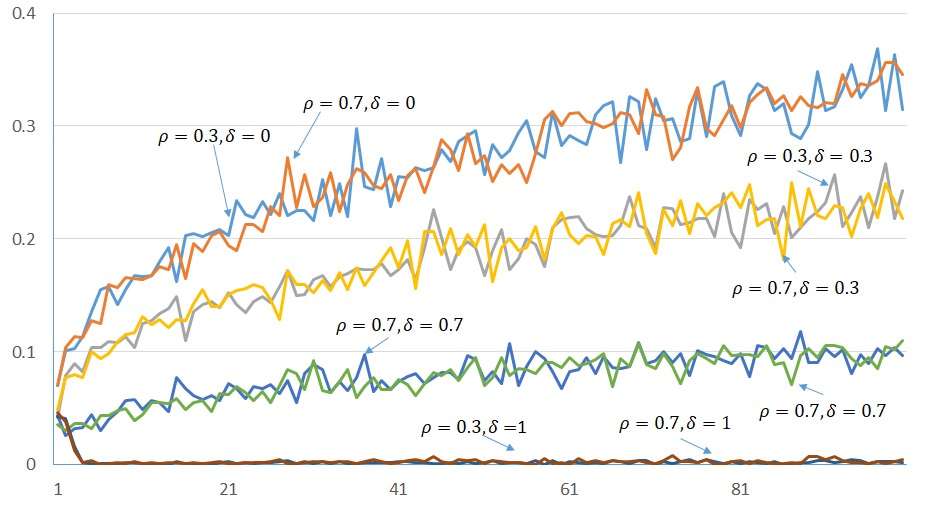}
	\caption{\; Failure intensity curves generated based on the failure intensity functions in Eq. \eqref{eq:MyModel-0}.}\label{fig:MyModel}
\end{figure}

\begin{figure}[!h]
	\includegraphics[width=0.6\linewidth]{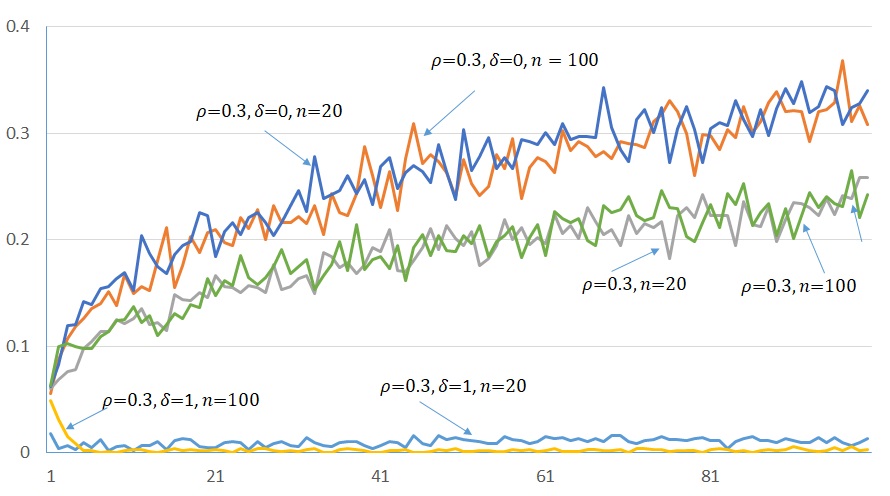}
	\caption{\; The curves with $\rho=0.3$.}\label{fig:SGRP_MyModel1}
\end{figure}
\begin{figure}[!ht]
	\includegraphics[width=0.6\linewidth]{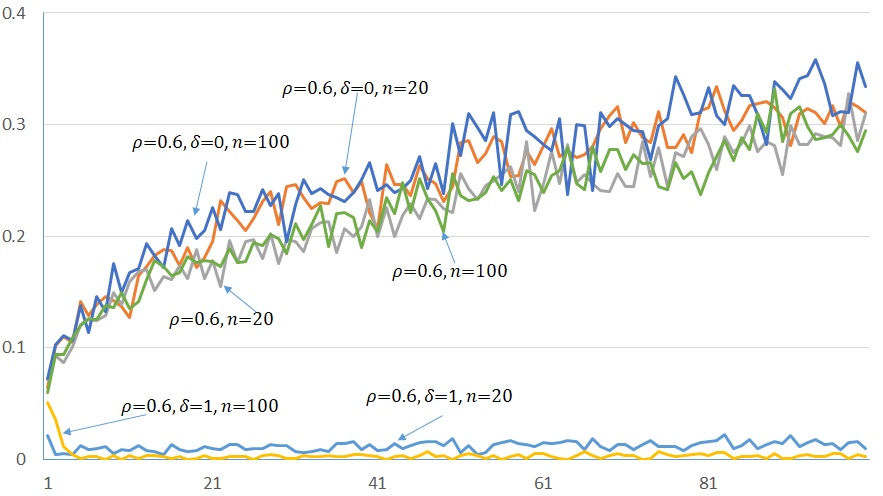}
	\caption{\; The curves with $\rho=0.6$.}\label{fig:SGRP_MyModel2}
\end{figure}

\section{Conclusions}
\label{sec:conclusions}
This paper derived the lower and upper bounds of the rate of the superimposed renewal process (SRP) and then extended to the superposition of the generalized renewal processes (SGRP). It proposed a linear combination of the bounds to approximate the SGRP. The method of simulating SGRP was also proposed. 

The weighting factor (i.e., $\delta$) was assumed to be deterministic. A possible extension is to assume it to be a random variable with its probability distribution supported on a bounded interval. This will be investigated in our future work.

\end{document}